\documentclass{llncs}

\usepackage{localmacros} 

\title{Semantic A-translations and Super-consistency entail Classical
  Cut Elimination} \author{Lisa Allali\inst{1} \and 
  Olivier Hermant \inst{2}} \institute{\'{E}cole Polytechnique, INRIA \&
  R\'egion Ile de France \email{allali@lix.polytechnique.fr} \and CRI,
  MINES ParisTech \email{olivier.hermant@mines-paristech.fr}}

\begin{document}

\maketitle

\begin{abstract}
We show that if a theory R defined by a rewrite system is
super-consistent, the classical sequent calculus modulo R enjoys the
cut elimination property, which was an open question. For such
theories it was already known that proofs strongly normalize in
natural deduction modulo R, and that cut elimination holds in the
intuitionistic sequent calculus modulo R.

We first define a syntactic and a semantic version of Friedman's
A-translation, showing that it preserves the structure of pseudo-Heyting
algebra, our semantic framework. Then we relate the interpretation of
a theory in the A-translated algebra and its A-translation in the
original algebra. This allows to show the stability of the
super-consistency criterion and the cut elimination theorem.\\

\end{abstract}

\section{Introduction} \label{sec:intro}
Deduction Modulo is a formalism that aims at separating computation
from reasoning in proofs by making inferences {\em modulo} some
congruence. This congruence is generated by rewrite rules on terms and
on propositions, and, assuming confluence and
termination, it is decidable by blind computation (normalization).

Rewrite rules on propositions is a key feature, allowing to
express in a first-order setting {\em without} any axiom theories such
as higher-order logic \cite{HOLls,PNM} or arithmetic
\cite{ARITH}. Reasoning without axioms turns out to be a critical
advantage for automated theorem provers \cite{JA12,RB04,BH06a,GBur10}
to not get lost during proof-search.

As a counterpart, fundamental properties such as cut elimination
become a hard challenge. At the same time it is needed at both
theoretical (consistency issues, e.g.) and practical levels, for
instance to ensure the completeness of the proof-search
algorithm of the aforementioned theorem provers. In the general case,
it does not hold and this is why new techniques have been developed in
order to ensure cut elimination for the widest possible range of
rewrite systems.

Anticipating the definitions of \mysec{\ref{sec:defs}}, let us
give two examples (see also \mysec{\ref{sec:sc-stab}}) to illustrate
the failure of cut elimination and/or normalization in general. For
terminating (and confluent) examples, see \cite{OHer05a}. The
congruence generated by the rewrite system $P \rew P \limp Q$ enables
to prove the sequent $\vdash Q$ with a cut and this proof is neither
normalizable in Natural deduction (the $\lambda$-term $(\lambda
x.x~x)~(\lambda x.x~x)$, that represents the aforementioned proof is
typable) nor admits cut \cite{PNM}. Instantiating $Q$ by $P$ yields
the rewrite system $P \rew P \limp P$. This allows for the same
non-normalizing proof, while $\vdash A$ becomes provable in only two
steps and without cut ; more generally, semantic  means \cite{OHer05a}
show that in this case cut is admissible, showing the independence of
normalization and  cut elimination. All those questions are
undecidable \cite{GBurCKir10}.\\

A first path to solve this problem, investigated in \cite{PNM}, is to
show that a congruence has a reducibility candidate-valued model. Then
any proof normalizes in natural deduction modulo this congruence. This
propagates to cut elimination in intuitionistic sequent calculus
modulo, but fails to directly extend to classical sequent calculus
modulo. To fix this, a second derived criterion is proposed.

A second way is {\em super-consistency}, a notion developed in
\cite{TVA} that is a semantic criterion independent from reducibility
candidates. It assumes the existence, for a given congruence, of a
model for {\em any} pseudo-Heyting algebra. Since the reducibility
candidates model of \cite{PNM} is an instance of pseudo-Heyting
algebra, this criterion implies that of \cite{PNM}, and all of its
normalization / cut elimination corollaries. So this suffers the same
drawback. A recent work \cite{ABruOHerCHou11} has also extended the
criterion to the classical case, but still requires a modification of
the criterion - specifically, pseudo-Heyting algebras become
pre-Boolean algebras.

The beauty of super-consistency is that it is not hardwired for a
particular deduction system. That is why it should indifferently prove 
cut elimination for the natural deduction, the intuitionistic
as well as the classical sequent calculus. This is exactly what
show here: cut-elimination for the classical sequent calculus modulo
a given congruence, assuming the {\em unmodified} congruence has the
{\em unmodified} super-consistency property.\\

After giving the definitions one would need to keep the paper as much
self contained as possible, we introduce shortly the deduction modulo,
relying on a basic knowledge first-order logic. Then we present the
$A$-translation of propositions and rewrite systems \cite{PNM},
inspired by Friedman's $A$-translation \cite{FRIEDMAN}, a refinement
of double-negation translations, that bridges the intuitionistic and
the classical worlds.

The core of the paper resides in the lifting of this translation on
pseudo-Heyting algebras, at the semantic level. After verifying that
all properties are preserved, we show that super-consistency is stable
by $A$-translation: the rewrite system has a model in the translated
algebra, so the translated rewrite system has a model in the original
algebra.

Those results allow us to deduce that super-consistency is
sufficient to prove cut-elimination in classical sequent calculus,
propagating the normalization property of natural deduction modulo to
cut elimination in intuitionistic and eventually classical sequent
calculus, following \cite{PNM}.

\section{Definitions} \label{sec:defs}

\subsection{Pseudo-Heyting Algebra}

\begin{definition}[pseudo-Heyting algebra (pHA,
    \cite{TVA})]\label{def:pHA}
Let $\calB$ be a set and $\leq$ a relation on it, $\calA$ and $\calE$
be subsets of $\wp(\calB)$, $\stop$ and $\sbot$ be elements of
$\calB$, $\simp$, $\sand$, and $\sor$ be functions from
$\calB \times \calB$ to $\calB$, $\sfa$ be a function from $\calA$
to $\calB$ and $\sex$ be a function from $\calE$ to $\calB$.  The
structure $\tildeB = \langle \calB, \leq, \calA, \calE, \stop,
\sbot, \simp, \sand, \sor, \sfa, \sex\rangle$
is said to be a {\em pseudo-Heyting algebra} if for
all $a$, $b$, $c$ in $\calB$, $A$ in $\calA$ and $E$ in $\calE$:
\begin{enumerate}
\item $a \leq a$ and if $a \leq b$, $b \leq c$ then $a \leq c$
  ($\leq$ is a pre-order), 
\item $a \leq \stop$ and $\sbot \leq a$ (maximum and minimum element),
\item $a \sand b \leq a$, $a \sand b \leq b$ and if $c
  \leq a$, $c \leq b$ then $c \leq a \sand b$,
\item $a \leq a \sor b$, $b \leq a \sor b$ and if $a \leq c$, $b \leq
  c$ then $a \sor b \leq c$,
\item for any $x \in A$, $\sfa A \leq x$ and if for
  any $x \in A$, $b \leq x$ then $b \leq \sfa A$,
\item for any $x \in E$, $x \leq \sex E$ and if for any $x \in E$,
  $x \leq b$ then $\sex E \leq b$,
\item $a \leq b \simp c$ iff $a \sand b \leq c$.
\end{enumerate}
\end{definition}

Axioms for $\sand$ and $\sfa$ (resp. $\sor$ and $\sex$) confer them
the property of a greatest lower bound (resp. lowest upper bound),
while the unicity of the latters is not guaranteed, since $\leq$ is
not anntisymmetric. Another guise of pHAs are Truth Value Algebras
\cite{TVA}. Also, $\sand$ and $\sor$ are
easily shown to be pre-commutative ($a \sand b \leq\geq b \sand a$)
and pre-associative.

\begin{definition}[Full \cite{TVA}]\label{def:full}
A pseudo-Heyting algebra is said to be {\em full} if $\calA = \calE =
\wp(\calB)$, {\em i.e.} if $\sfa~A$ and $\sex~A$ are defined for all
$A \subset \calB$.
\end{definition}

In this paper, all the pHA considered are full. When the pre-order
is antisymmetric, then a full pHA is exactly a complete HA, in the
terminology of \cite{ATroDvDal88}. In this paper, complete refers to
the order $\eqK$ described below.

\begin{definition}[Ordered pseudo-Heyting algebra]\label{def:order}
A pseudo-Heyting algebra $\tildeB$ is called {\em ordered} if it is
equipped with an additional {\em order} relation $\eqK$ on
$\calB$ such that 
\begin{itemize}
\item $\eqK$ is a refinement of $\leq$, {\em i.e.} if $a \eqK b$ then
  $a \leq b$,
\item $\stop$ is a maximal element,
\item $\sand$, $\sor$, $\sfa$ and $\sex$ are
  monotonous, $\simp$ is left anti-monotonous and right
  monotonous.
\end{itemize}
\end{definition}

\mydef{\ref{def:order}} is an adapation to pHA of the corresponding 
definition of \cite{TVA}. The ``refinement condition'' is shown in
\cite{TVA} to be a derived property (\myprop{4}), but it is in fact
trivially equivalent to the closure condition of $\Bplus$.

\begin{definition}[Complete ordered pseudo-Heyting algebra
    \cite{TVA}] \label{def:complete} 
An ordered pseudo-Heyting algebra $\tildeB$ is said to be {\em complete}
if every subset of $\calB$ has a greatest lower bound for $\eqK$.
Notice that this implies that every subset also has a least upper
bound. We write $glb(a,b)$ and $lub(a,b)$ the greatest lower bound and
the least upper bound of $a$ and $b$ for the order $\eqK$.
\end{definition}

The order relation $\eqK$ does not define a Heyting algebra order and,
if by chance it does, the Heyting algebra operations may be different
from those of $\tildeB$ .



\subsection{Rewrite System}

We work in usual predicate logic. Terms are variables and applied
function symbols along their arity. Propositions are atoms
(applied predicate symbols along their arity), and compound
propositions with the help of connectives $\land, \lor, \limp, \top, 
\bot$ and quantifiers $\lfa, \lex$. $\alpha$-equivalent propositions
are identified. To avoid parenthesis, $\limp$ and $\simp$
are considered to be {\em left} associative, therefore $A \limp B
\limp B$ reads $(A \limp B) \limp B$. Negation is not a primitive
connective, and is defined by $A \limp \bot$.

\begin{definition}[Proposition rewrite rule]
We call {\em proposition rewrite rule} any rule $P \rew A$ rewriting
atomic propositions $P$ into an arbitrary proposition $A$ such that
$\FV(A) \subseteq \FV(P)$.
\end{definition}


\begin{definition}[Proposition rewrite system]
We define a {\em proposition rewrite system} as an orthogonal
\cite{TeReSe03}, hence confluent, set of proposition rewrite rules. The
congruence generated by this rewrite system is noted $\equiv$.
\end{definition}

\subsection{Interpretation}


\begin{definition}[$\tildeB$-valued structure \cite{TVA}]\label{def:structure}
Let $\calL = \langle f_i, P_j \rangle$ be a language in predicate
logic and $\tildeB$ be a pHA, a {\em $\tildeB$-valued
  structure} $\calM = \langle \calM, \tildeB, \hat{f}_i, \hat{P}_j
\rangle$ for the language $\calL$ is a structure such that 
$\hat{f_i}$ is a function from $\calM^n$ to $\calM$ where $n$ is the
arity of the symbol $f_i$ and $\hat{P_j}$ is a function from $\calM^n$
to $\calB$, the domain of $\tildeB$, where $n$ is the arity of the
symbol $P_i$.
\end{definition}

\begin{definition}[Denotation \cite{TVA}]
Let $\tildeB$ be a pHA, $\calM$ be a $\tildeB$-valued
structure and $\phi$ be an assignment, i.e. a function associating
elements of $\calM$ to variables. The denotation in $\calM$ of a
proposition $A$ or of a term $t$ is defined as:

\begin{itemize}
\item $\interp{ x } = \phi(x)$,
\item $\interp{ f(t_1, ..., t_n) } = \hat{f}(\interp{ t_1}, ...,
  \interp{t_n})$,
\item $\interp{ P(t_1, ..., t_n) } = \hat{P}(\interp{ t_1}, ...,
  \interp{ t_n})$,
\item $\interp{ \top } = \stop$,
\item $\interp{ \bot } = \sbot$,
\item $\interp{ A \Rightarrow B } = \interp{ A } ~\simp~ \interp{
  B }$,
\item
$\interp{ A \wedge B } = \interp{ A } ~\sand~ \interp{ B }$,
\item
$\interp{ A \vee B } = \interp{ A } ~\sor~ \interp{ B }$,
\item $\interp{ \lfa x~A } = \sfa~\{ \interpbis{ A }~|~e \in \calM\}$,
\item $\interp{ \lex x~A } = \sex~\{ \interpbis{ A }~|~e \in \calM\}$.
\end{itemize}
The denotation of a proposition containing quantifiers is always
defined if the pHA is full, otherwise it may be undefined.
\end{definition}

\begin{definition}[Model \cite{TVA}]
The $\tildeB$-valued structure $\calM$ is said to be {\em a model of}
a rewrite system $R$ if for any two propositions $A, B$ such that $A
\equiv B$, $\interpinb{A} = \interpinb{B}$.
\end{definition}

Soundness and completeness hold \cite{TVA}: the sequent $\Gamma \vdash
B$ is provable if and only if $\interpinb{\G} \leq \interpinb{B}$ for
any pseudo-Heyting algebra $\tildeB$ and any model interpretation for
$R$ in $\tildeB$.
The direct way is an usual induction \cite{TVA}, while the converse
is a direct consequence of the completeness theorem with respect to
Heyting algebra. For instance one can construct the Lindenbaum
algebra \cite{TVA}, or a context-based algebra \cite{OHerJLip08b}. 

\subsection{Classical Sequent Calculus Modulo}

\myfig{\ref{fig:seq-cal}} recalls the classical sequent calculus
modulo. It depends on a congruence $\equiv$ determined by a fixed
rewrite system $R$. If $R$ is empty $\equiv$ boils down to
syntactic equality and we get usual sequent calculus. The
intuitionistic sequent calculus modulo has the same rules,
except that the right-hand sides of sequents contain at most {\em one}
proposition. Two rules are impacted: $\lor$-r splits into two rules
$\lor_1$ and $\lor_2$, and, in the right premiss of the $\limp$-left
rule, $\Delta$ is overwritten by $A$.

\begin{center}
\begin{figure}[htbp]
\scriptsize
\begin{tabular}{|rcl|}
\hline & &  \\
\multicolumn{3}{|c|}{\bf{ identity group }}\\
\AXC{}
\LeftLabel{{axiom}, $A \equiv B$}
\UICm{A \vdash B}
\DP
& ~~~~~~ & 
\AXCm{\G \vdash A, \D}
\AXCm{\G, B \vdash \D}
\RightLabel{{cut}, $A \equiv B$}
\BICm{\G \vdash \D}
\DP\\& ~~~~~~ &\\
\multicolumn{3}{|c|}{\bf{ logical group }}\\
\AXCm{\Gamma, A, B \vdash \Delta}
\LeftLabel{{$\land$-l}, $C \equiv A \land B$}
\UICm{\Gamma, C \vdash \Delta}
\DP
& & 
\AXCm{\Gamma \vdash A, \Delta}
\AXCm{\Gamma \vdash B, \Delta}
\RightLabel{{$\land$-r}, $C \equiv A \land B$~~}
\BICm{\Gamma \vdash C, \Delta}
\DP
\\& &\\
\AXCm{\Gamma, A \vdash \Delta}
\AXCm{\Gamma, B \vdash \Delta}
\LeftLabel{{$\lor$-l}, $C \equiv A \lor B$}
\BICm{\Gamma, C \vdash \Delta}
\DP
& & 
\AXCm{\Gamma \vdash A, B, \Delta}
\RightLabel{{$\lor$-r}, $C \equiv A \lor B$}
\UICm{\Gamma \vdash C, \Delta}
\DP
\\& &\\
\AXCm{\Gamma, B \vdash \Delta}
\AXCm{\Gamma \vdash A, \Delta}
\LeftLabel{~~{$\limp$-l}, $C \equiv A \limp B$}
\BICm{\Gamma, C \vdash \Delta}
\DP
& & 
\AXCm{\Gamma, A \vdash B, \Delta}
\RightLabel{{$\limp$-r}, $C \equiv A \limp B$}
\UICm{\Gamma \vdash C, \Delta}
\DP
\\& &\\
\AXCm{}
\LeftLabel{{$\bot$-l}, $A \equiv \bot$}
\UICm{A \vdash}
\DP
& & 
\AXCm{}
\RightLabel{{$\top$-r}, $A \equiv \top$}
\UICm{\vdash A}
\DP
\\& &\\
\AXCm{\Gamma, \{t/x\}A \vdash \Delta}
\LeftLabel{{$\lfa$-l}, $B \equiv \lfa{x} A$}
\UICm{\Gamma, B \vdash \Delta}
\DP
& & 
\AXCm{\Gamma \vdash A, \Delta}
\RightLabel{{$\lfa$-r}, $B \equiv \lfa{x} A$, $x$ fresh}
\UICm{\Gamma \vdash B, \Delta}
\DP
\\& &\\
\AXCm{\Gamma, A \vdash \Delta}
\LeftLabel{~~~{$\lex$-l},  $B \equiv \lex{x} A$, $x$ fresh}
\UICm{\Gamma, B \vdash \Delta}
\DP
& & 
\AXCm{\Gamma \vdash \{t/x\}A, \Delta}
\RightLabel{{$\lex$-r}, $B \equiv  \lex{x} A$}
\UICm{\Gamma \vdash B, \Delta}
\DP
\\& &\\
\multicolumn{3}{|c|}{\bf{ structural group }}\\
\AXCm{\Gamma, B_1, B_2 \vdash \Delta}
\LeftLabel{{contr-l}, $A \equiv B_1 \equiv B_2$}
\UICm{\Gamma, A \vdash \Delta}
\DP
& & 
\AXCm{\Gamma \vdash B_1, B_2, \Delta}
\RightLabel{{contr-r}, $A \equiv B_1 \equiv B_2$}
\UICm{\Gamma \vdash A, \Delta}
\DP
\\& &\\
\AXCm{\Gamma \vdash \Delta}
\LeftLabel{{weak-l}}
\UICm{\Gamma, A \vdash \Delta}
\DP
& & 
\AXCm{\Gamma \vdash \Delta}
\RightLabel{{weak-r}}
\UICm{\Gamma \vdash A, \Delta}
\DP\\& &\\
\hline
\end{tabular}
\caption{Classical sequent calculus modulo}
\label{fig:seq-cal}
\end{figure}
\end{center}

\subsection{Super-consistency}

\begin{definition}[Super-consistency \cite{TVA}] \label{def:sc} 
A rewrite system $R$ (a congruence $\equiv$) in deduction modulo is 
{\em super-consistent} if it has a $\tildeB$-valued model for all
full, ordered and complete pseudo-Heyting algebra $\tildeB$.
\end{definition}

Super-consistency is akin to consistency with respect to {\em all}
pHA. Note that the choice of the structure
(\mydef{\ref{def:structure}}) is open. Considering only HA is not
enough, as the rewrite system $P \rew P \limp P$ devised in 
\mysec{\ref{sec:intro}} , as well as the one of 
\mysec{\ref{sec:sc-stab}} would then be super-consistent but not
normalizing.

\section{A-translations}

Instead of first performing a negative translation \cite{PNM} and
then the proper $A$-translation, as in the original work of Friedman
\cite{FRIEDMAN}, we consider a variant of the composition of both.

\subsection{Syntactic Translation of a Proposition}

\begin{definition}[A-translation of a proposition]\label{def:atrans-prop}

Let $B$ be a proposition. Let $A$ be a proposition in which free 
variables are not bound by quantifiers in $B$. $A$ is said
$B$-unbound. We let $\aprop{B}{A}$ be:
\begin{itemize} 
\item $\aprop{B}{A}=B$ if $B$ is atomic,
\item $\aprop{\top}{A}=\top$,
\item $\aprop{\bot}{A}=\bot$,
\item $\aprop{(B \limp C)}{A} = (\aprop{B}{A} \dimpa) \limp
  (\aprop{C}{A} \dimpa)$,
\item $\aprop{(B\wedge C)}{A}=(\aprop{B}{A} \dimpa) \wedge (\aprop{C}{A}
  \dimpa)$,
\item $\aprop{(B \vee C)}{A} = (\aprop{B}{A} \dimpa) \vee (\aprop{C}{A}
  \dimpa)$,
\item $\aprop{(\lfa x~B)}{A}=\lfa x~(\aprop{B}{A} \dimpa)$,
\item $\aprop{(\lex x~B)}{A}=\lex x~(\aprop{B}{A} \dimpa)$.
\end{itemize}
\end{definition}
\begin{remark}
Kolmogorov's double negation translation \cite{PNM} of $B$ is
$\neg\neg \aprop{B}{\bot}$. As well as this translation has been
simplified by G\"odel, Gentzen and others \cite{GOD,GEN,ATroDvDal88},
we can also simplify \mydef{\ref{def:atrans-prop}} so that it
introduces less $A$.

\end{remark}

\begin{definition}[A-translation of a rewrite system]
Let $R = \{ P_i \rew A_i \}$ be a proposition rewrite system and $A$
be a formula that is $A_i$-unbound for all $i$. We define its
$A$-translation, written $\aprop{R}{A}$, as $\{ P_i \rew
\aprop{A_i}{A} \}$.
\end{definition}

\subsection{Semantic $a$-translation of a pHA}

We now lift the $A$-translation process at the semantic level.

\begin{definition}[Semantic $a$-translation] \label{def:a-trans-sem}
Let $\tildeB$ be the full pseudo-Heyting algebra $\langle \calB, \leq,
\wp(\calB), \wp(\calB), \stop, \sbot, \simp, \sand, \sor, \sfa,
\sex\rangle$ and let $a \in \calB$.

We let $\asem{\tildeB}{a}$ be the structure $\langle \calB, 
\aleq, \wp(\calB), \wp(\calB), \atop, \abot, \aimp, \aand,
\aor, \afa, \aex \rangle$, that we call the $a$-translation of
$\tildeB$, where:
\begin{itemize}
\item $b \aleq c$ iff $b \dsimpa \leq c \dsimpa$,
\item $\atop \triangleq \stop$,
\item $\abot \triangleq \sbot$,
\item $b \aimp c \triangleq ((b \dsimpa) \simp (c \dsimpa))$,
\item $b \aand c \triangleq ((b \dsimpa) \sand (c \dsimpa))$,
\item $b \aor c \triangleq  ((b \dsimpa) \sor (c \dsimpa))$,
\item $\afa A \triangleq ( \sfa~ (A \dsimpa))$,
\item $\aex A \triangleq (\sex~(A \dsimpa))$.
\end{itemize}
with the convention that, for any $A \subseteq \calB$, $A \dsimpa = \{
b \dsimpa ~|~ b \in A\}$.
\end{definition}

We may straightforwardly check that  $\langle \calB, 
\aleq, \wp(\calB), \wp(\calB), \atop, \abot, \aimp, \aand,
\aor, \afa, \aex \rangle$ is a valid {\em structure}, in the sense that
$\aleq$ operators are well-defined; in particular $\afa$ and $\aex$
are defined for any subset of $\calB$. We show below that it is also a
full, ordered and complete pHA.

\section{Results}

\subsection{On the $a$-translation of a pHA}

We recall some useful facts about the semantic implication that hold
in pseudo-Heyting algebras:
\begin{proposition} \label{prop:eq-imp}
Let $\calB$ be a pHA and $a, b, c \in \tildeB$ such that $b \leq
c$. Then:
\begin{eqnarray}
b & \leq & a \simp b \label{eq:imp-1}\\
a \simp b \sand a & \leq & b \label{eq:imp-2}\\
b & \leq & b \simp a \simp a \label{eq:imp-5}\\
a \simp b & \leq & a \simp c \label{eq:imp-3}\\
c \simp a & \leq & b \simp a \label{eq:imp-4}\\
b \dsimpa & \leq & c \dsimpa \label{eq:imp-7}\\
b \simp a \dsimpa & \leq & b \simp a \label{eq:imp-6}
\end{eqnarray}
\end{proposition}
\begin{proof}
Standard, using the definition of $\simp$. Let us show \ref{eq:imp-6}:
by \ref{eq:imp-5} $b \leq b \dsimpa \leq b \dsimpa \dsimpa$. Then by
definition of $\simp$ we get first $(b \simp a \dsimpa) \sand b \leq
a$ and then $b \simp a \dsimpa \leq b \simp a$. \qed{}
\end{proof}

\begin{proposition}\label{BTVA}
If $\tildeB$ is a full pHA then its $a$-translation $\aB$ is a full
pHA.
\end{proposition}
\begin{proof}
We check one by one all the points of \mydef{\ref{def:pHA}} and
\mydef{\ref{def:full}}:
\begin{itemize}
\item $\aleq$ is a pre-order: inherited from $\leq$
\item $b \aleq \atop$ since $b \dsimpa \leq \stop \dsimpa$ (by
  \ref{eq:imp-7}). Similarly for $\abot$.
\item $b \aand c$ is a lower bound of $b$ and $c$. Let us
  show $b \aand c \aleq b$. By definition of $\aand$ and of $\sand$,
  $b \aand c \leq b \dsimpa$. By \ref{eq:imp-7} $(b \aand c) \dsimpa
  \leq b \dsimpa \dsimpa$ and by \ref{eq:imp-6} of
  \myprop{\ref{prop:eq-imp}} $b \dsimpa \dsimpa \leq b \dsimpa$ which
  allows us to conclude. Similar arguments show that $b \aand c \aleq
  c$.

\item $b \aand c$ is a greatest lower bound of $b$ and $c$: let $d$
  such that $d\aleq b$ and $d \aleq c$. By definition of $\sand$,
  $\aand$ and of $\aleq$, $d \dsimpa \leq b \aand c$ and by
  \ref{eq:imp-5} of \myprop{\ref{prop:eq-imp}}, $b \aand c \leq (b \aand
  c) \dsimpa$ which allows us to conclude.

\item $b \aor c$ is an upper bound of $b$ and $c$. Let us show $b
  \aleq b \aor c$. By definition of $\sor$ and of $\aor$, $b \dsimpa
  \leq b \aor c$. We conclude by \ref{eq:imp-5} of
  \myprop{\ref{prop:eq-imp}}. Similar arguments show that $c \aleq b \aor
  c$.

\item $b \aor c$ is a least upper bound of $b$ and $c$. Let $d$ such
  that $b \aleq d$ and $c \aleq d$. Then, $(b \dsimpa) \sor (c
  \dsimpa) \leq d \dsimpa$ and by \ref{eq:imp-7} of
  \myprop{\ref{prop:eq-imp}}, $((b \dsimpa) \sor (c \dsimpa)) \dsimpa
  \leq d \dsimpa \dsimpa$. By applying \ref{eq:imp-6}, $ d \dsimpa
  \dsimpa \leq d \dsimpa$, which allows us to conclude.

\item $\afa A$ is a lower bound of $A$. Let $x \in A$. Then $\afa A
  \leq x \dsimpa$ by definition of $\afa$ and $\sfa$. Using
  \myprop{\ref{prop:eq-imp}}, by \ref{eq:imp-7} $(\afa A) \dsimpa \leq
  x \dsimpa \dsimpa$ and by \ref{eq:imp-6} $x \dsimpa \dsimpa \leq x
  \dsimpa$, which allows us to conclude.

\item $\afa A$ is a greatest lower bound of $A$. Let $b$ such that
  for any $x \in A$, $b \aleq x$. Then $b \dsimpa \leq x \dsimpa$ and
  by definition of $\sfa$, $b \dsimpa \leq \sfa (A \dsimpa) = \afa A$. By
  \ref{eq:imp-5} of \myprop{\ref{prop:eq-imp}}, $\afa A \leq (\afa A)
  \dsimpa$, which allows us to conclude.

\item $\aex A$ is an upper bound of $A$. Let $x \in A$. Then $x
  \dsimpa \leq \aex A$ by definition of $\aex$ and $\sex$. By
  \ref{eq:imp-5} $\aex A \leq (\aex A) \dsimpa$, which allows us to
  conclude.

\item $\aex A$ is a least upper bound of $A$. Let $b$ such that for
  any $x \in A$, $x \aleq b$. Then $x \dsimpa \leq b \dsimpa$ and
  by definition of $\sex$, $\aex A = \sex (A \dsimpa) \leq b
  \dsimpa$. By \myprop{\ref{prop:eq-imp}} we derive $(\aex A) \dsimpa
  \leq b \dsimpa \dsimpa$ and $b \dsimpa \dsimpa \leq b \dsimpa$,
  which allows us to conclude.

\item direct way of the implication property. Assume $b \aleq c \aimp
  d$, that is to say $b \dsimpa \leq ((c 
  \dsimpa) \simp (d \dsimpa)) \dsimpa$. As an intermediate result we
  claim that for any $x$, $y$ and $z$, $(x \simp (y \simp 
  z)) \simp z \simp a \leq  x \simp (y \simp a)$.

$$
\begin{array}{r@{\leq}ll}
x \simp (y \simp z) & x \simp (y \simp z) & 
\mbox { (reflexivity)}\\
(x \simp (y \simp z)) \sand x \sand y & z & 
\mbox{ (Definition of $\simp\!\!$)}\\
x \sand y & x \simp (y \simp z) \simp z & 
\mbox{ (Definition of $\simp\!\!$)}\\
x \sand y & [x \simp (y \simp z) \simp z] \dsimpa\!\! & 
\mbox{ (\myprop{\ref{prop:eq-imp}})}\\
\left[ x \simp (y \simp z) \simp z \simp a \right] \sand x \sand y & a &
\mbox{ (Definition of $\simp\!\!$)} \\
x \simp (y \simp z) \simp z \simp a  & x \simp (y \simp a) & 
\mbox{ (Definition of $\simp\!\!$)}\\
\end{array}
$$

If we replace in this last inequality $x$ by $c \dsimpa$, $y$ by $d
\simp a$ and $z$ by $a$, we get $((c \dsimpa) \simp (d \dsimpa))
\dsimpa \leq ((c \dsimpa) \simp (d \dsimpa))$ so that we derive $b
\dsimpa \leq (c \dsimpa) \simp (d \dsimpa)$, or said otherwise $(b
\dsimpa) \sand (c \dsimpa) \leq d \dsimpa$. By
\myprop{\ref{prop:eq-imp}} we get the inequality $((b \dsimpa) \sand
(c \dsimpa)) \dsimpa \leq d \dsimpa \dsimpa \leq d \dsimpa$, which is
exactly $b \aand c \aleq d$.

\item conversely, assume $b \aand c \aleq d$, i.e. $((b \dsimpa) \sand (c
  \dsimpa)) \dsimpa \leq d \dsimpa$. By \ref{eq:imp-7} of
  \myprop{\ref{prop:eq-imp}} we get that
  $((b \dsimpa) \sand (c \dsimpa)) \leq ((b \dsimpa) \sand (c
  \dsimpa)) \dsimpa$, so $b \dsimpa \!\leq\! (c \dsimpa) \simp (d
  \dsimpa)$ by definition of $\simp$. And by \ref{eq:imp-5} we get that $(c
  \dsimpa) \simp (d \dsimpa)\! \leq ((c \dsimpa) \simp (d \dsimpa))
  \dsimpa$, which allows us to conclude. \qed{}
\end{itemize}
\end{proof}

\begin{proposition}\label{prop:full-ordered}
Let $\calB$ be a full and ordered pHA, with respect to $\eqK$. Let $a \in
\tildeB$. The $a$-translation $\asem{\calB}{a}$ of $\calB$ is a full
and ordered pHA with respect to $\eqK$.
\end{proposition}
\begin{proof}
By \myprop{\ref{BTVA}}, $\asem{\calB}{a}$ is a full pHA. We check
\mydef{\ref{def:order}}:
\begin{itemize}
\item $\eqK$ is by definition an order relation on $\calB$, which is
  also the domain of $\asem{\tildeB}{a}$.
\item $\atop$ (resp. $\abot$) is maximal (resp. minimal) for
  the same reason.
\item assume $b \eqK c$. Then $b \leq c$ and by
  \myprop{\ref{prop:eq-imp}} $b \aleq c$.
\item $\aand$ is monotonous. Let $b, c, d$ be elements of the algebra,
  and assume $b \eqK c$. By left-antimonotonousity of $\eqK$ with respect
  to $\simp$, $b \dsimpa \eqK c \dsimpa$, so $b \aand
  d = (b \dsimpa) \sand (d \dsimpa) \eqK (c \dsimpa) \sand (d \dsimpa)
  =  c \aand d$ by monotonicity of $\eqK$ with respect to $\sand$.
\item the other properties with respect to $\aor$, $\aimp$, $\afa$ and
  $\aex$ are shown in the same way:  first notice that $b \dsimpa \eqK
  c \dsimpa$ and then use the corresponding property of $\eqK$ with
  respect to the original connective. Remember that, for $A, A'$ sets of
  elements of $\asem{\tildeB}{a}$, $A \eqK A'$ means that, for any $x
  \in A$, there exists $y \in A'$ such that $x \eqK y$. \qed{}
\end{itemize}
\end{proof}

\begin{proposition}\label{prop:complete}
If $\calB$ is a full, ordered and complete pHA, then its
$a$-translation $\asem{\calB}{a}$ is a full, ordered and complete pHA.
\end{proposition}
\begin{proof}
From \myprop{\ref{prop:full-ordered}}, $\asem{\tildeB}{a}$ is full and 
ordered. The greatest lower and lowest upper bounds of any
$A$ subset of $\calB$ (the domain of $\asem{\tildeB}{a}$) for $\eqK$
are members of $\calB$ because $\tildeB$ is complete. The condition of
\mydef{\ref{def:complete}} is fulfilled. \qed{} 
\end{proof}

\subsection{Relating Interpretations}

\begin{proposition}\label{prop:propag-1}
Let $\calB$ be a full, ordered and complete pHA. Consider a
$\tildeB$-valued structure $\calM$ and note $\interpinb{.}$ the
denotation $\calM$ generates in $\tildeB$. Let $A$ be a closed
proposition and let $\asem{B}{\interpinb{A}}$ be the
$\interpinb{A}$-translation of $\calB$:
\begin{enumerate}
\item $\calM$ is also a $\asem{\tildeB}{\interpinb{A}}$-valued
  structure. Let $\interpinAb{.}$ be the denotation it generates in
  $\asem{\tildeB}{\interpinb{A}}$.
\item for any term $t$, any assignment $\phi$, $\interpinb{t}_\phi =
  \interpinAb{t}_\phi$.
\item For any proposition $B$, any assignment $\phi$,
  $\interpinb{\aprop{B}{A}}_\phi = \interpinAb{B}_\phi$.
\end{enumerate}
\end{proposition}

$A$ is chosen to be closed, otherwise we would need to consider
$\interpinb{A}_{\phi_0}$ for a fixed $\phi_0$ and consider only
formul\ae{} $B$ such that $A$ is $B$-unbound. We rather avoid those
complications.

\begin{proof}
$\calM$ is obviously a $\asem{\tildeB}{\interpinb{A}}$-valued
  structure (see \mydef{\ref{def:structure}}) since the domain of both
  pHAs is the same and $\calM$ assigns values only to atomic
  constructs. The second claim is also obvious, since the domain for
  terms does not change. We prove the last claim by an easy induction
  on the structure of $B$, where we omit the valuation $\phi$, which
  plays no role. We note $a = \interpinb{A}$ in the definition of the
  operators of $\asem{\tildeB}{\interpinb{A}}$.
\begin{itemize}

\item if $B$ is an atomic formula $P(t_1, \cdots, t_n)$, then 
by construction and definition of the $A$-translation:
$$
\interpinb{\aprop{B}{A}} = \interpinb{B} = 
\hat{P}(\interpinb{t_1}, \cdots, \interpinb{t_n}) =
\hat{P}(\interpinAb{t_1}, \cdots, \interpinAb{t_n}) = 
\interpinAb{B}
$$ 

\item $\interpinb{\aprop{\top}{A}} = \stop = \interpinAb{\top}$,
  similarly for $\bot$.

\item $\interpinb{\aprop{(B \limp C)}{A}} = (\interpinb{\aprop{B}{A}}
  \simp 
  \interpinb{A} \simp \interpinb{A}) \simp (\interpinb{\aprop{C}{A}}
    \simp \interpinb{A} \simp \interpinb{A}) =
    \interpinb{\aprop{B}{A}} \aimp \interpinb{\aprop{C}{A}}$ which, by
    induction hypothesis is equal to $\interpinAb{B} \aimp
    \interpinAb{C} = \interpinAb{B \limp C}$.
\item similarily for $\land$ and $\lor$.
\item $\interpinb{\aprop{\lfa{x} B}{A}} = \sfa \{ \interpinb{
  \aprop{B}{A}}_{\langle x,d \rangle} \simp \interpinb{A} \simp
  \interpinb{A} ~|~ d \in \calM \}$ and by induction hypothesis and
  the notation of \mydef{\ref{def:a-trans-sem}}, this is equal to
  $\sfa  \{ \interpinAb{B}_{\langle x,d \rangle} ~|~ d \in \calM \}
  \simp \interpinb{A} \simp \interpinb{A} = \afa \{
  \interpinAb{B}_{\langle x,d \rangle} ~|~ d \in \calM \} =
  \interpinAb{\lfa{x} B}$.
\item similarly for $\lex$. \qed{}

\end{itemize}

\end{proof}

\subsection{Stability of Super-consistency} \label{sec:sc-stab} 

In this section we show that the super-consistency property of a
rewrite system is preserved by $A$-translation under certain conditions.\\

First, notice that the general statement is not true because nasty
interferences can happen if the $A$-translation is done with respect
to a $A$ containing propositions of the rewrite system. In
particular, we can lose the normalization property, which is implied
by super-consistency, and so, super-consistency itself. To illustrate
this, consider the following rewrite system consisting of the sole
rule $P \rew \top \land \top$. Super-consistency comes out easily:
given a pHA $\tildeB$, we let $\hat{P} = \stop \sand \stop$. But
super-consistency fails for its $P$-translated rewrite system:
$$
P \rew (\top \limp P \limp P) \land (\top \limp P \limp P)
$$

\begin{figure}
\scriptsize
\begin{tabular}{|c@{~}c|}
\hline & \\

\AXCm{\G \vdash \pi_1:~A}
\AXCm{\G \vdash \pi_2:~B}
\RightLabel{$\land_i, A \land B \equiv C$}
\LeftLabel{~}
\BICm{\G \vdash \langle \pi_1, \pi_2 \rangle:~C}
\DP

&

\AXCm{\G, x: A \vdash \pi:~B}
\RightLabel{$\limp_i, C \equiv A \limp B$}
\UICm{\G \vdash \lambda x. \pi:~C}
\DP 

\\ & \\

\AXCm{\G \vdash \pi:~C}
\RightLabel{$\land_{e1}, C \equiv A \land B$}
\UICm{\G \vdash fst(\pi):A}
\DP

&

\AXCm{\G \vdash \pi_1:~C}
\AXCm{\G \vdash \pi_2:~A}
\RightLabel{$\limp_e, C \equiv A \limp B~$}
\BICm{\G \vdash \pi_1~\pi_2:~B}
\DP

\\

 & \\

$fst \langle \pi_1, \pi_2 \rangle \rhd \pi_1$

&

$(\lambda x. \pi_1)~\pi_2 \rhd \{\pi_2/x\}\pi_1$
\\

& \\
\hline
\end{tabular}
\caption{Some typing and reduction rules of natural deduction modulo \cite{PNM}}
\label{fig:ded-nat}
\end{figure}

As we will see, in natural deduction we can define a proof-term that
is not normalizing. Adopting the syntax and typing rules of
\cite{PNM}, shown in \myfig{\ref{fig:ded-nat}}, we let $t_1$ and $t_2$
be the following $\lambda$-terms, $I$ being the constant corresponding
to the $\top$-intro rule:\footnote{At the price of readability, $I$ and
$\top$ can be everywhere safely replaced by $\lambda y. y$ and $B
  \limp B$, respectively.}
\begin{eqnarray*}
t_1 & = & \lambda x.[fst(x~I)~(\lambda z.(x~I))]\\
t_2 & = & \lambda z. \langle t_1, t_1 \rangle
\end{eqnarray*}
Those terms can be typed respectively by $\top \limp P \limp P$ and by
$\top \limp (\top \limp P \limp P \land \top \limp P \limp P)$ or,
using the congruence, by $\top \limp P$: both bound $z$ can be
assigned the type $\top$, while $x$ has the type $\top \limp P \equiv
\top \limp (\langle (\top \limp P) \limp P), (\top \limp P) \limp
P)\rangle$, this last type identification being the source of the
problems. With those terms, we form the following looping reduction
sequence:
\begin{eqnarray*}
t_1~t_2 & \rhd & fst(t_2~I)~(\lambda z.(t_2~I)) \\
        & \rhd & fst(\langle t_1, t_1 \rangle)~(\lambda z. \langle
t_1, t_1 \rangle)\\
        & \rhd & t_1~t_2
\end{eqnarray*}

\noindent Since we do not have normalization, we cannot have
super-consistency. This is why restricting $A$ is the key
to \mythm{\ref{thm:atrans-sc}}.

\begin{definition}[$R$-compatibility]
Let $R$ be a rewriting system. A proposition $A$ is said to be
$R$-compatible if and only if does not contain any predicate or
function symbol appearing in $R$.
\end{definition}

\begin{proposition}\label{prop:propag-2}
Let $R$ be a rewrite system, and $A$ be a closed proposition. Let
$\tildeB$ be a pHA and consider a $\tildeB$-valued structure
$\calM$, generating an interpretation $\interpinb{\_}$. Let
$\asem{\tildeB}{\interpinb{A}}$ be the $\interpinb{A}$-translation of
  $\tildeB$ and $\aprop{R}{A}$ be the $A$-translation of $R$.

If the interpretation $\interpinAb{\_}$ generated by $\calM$ in
$\asem{\tildeB}{\interpinb{A}}$ is a model of $R$ then $\aprop{R}{A}$
has a $\calB$-model.
\end{proposition}
\begin{proof}
Let $P \rew \aprop{F}{A} \in \aprop{R}{A}$. By hypothesis, $P \rew F
\in R$ and $\interpinAb{P} = \interpinAb{F}$.
We conclude by noticing that, by definition, $\interpinAb{P} =
\interpinb{P}$ and that, by \myprop{\ref{prop:propag-1}}, $\interpinAb{F} =
\interpinb{\aprop{F}{A}}$. \qed{}
\end{proof}

The main requirement of \myprop{\ref{prop:propag-2}} is that
$\interpinAb{\_}$ must be a model of $R$. The choice of $\interpinb{\_}$
is here a degree of freedom, but this is not sufficient, even
assuming super-consistency. Indeed, the example of the beginning of the
section shows that this is impossible if $A$ is not $R$-compatible. We
must go through the following definition lemma.

\begin{lemma}[Relative grafting of structures] \label{lem:struct-comp}
Let $\tildeB$ be a pHA and $\calM_0$ and $\calM_1$ be two
$\tildeB$-valued structures. Let $A$ be a proposition. We define
$\calM_2$, the $A$-grafting of $\calM_0$ onto $\calM_1$ as
the following $\tildeB$-structure:
\begin{itemize}
\item for any function symbol $f$, $\hat{f} = \hat{f}_0$ (the value
  assigned by $\calM_0$) if $f$ syntactically appears in $A$ and
  $\hat{f} = \hat{f}_1$ (the value assigned by $\calM_1$) otherwise.
\item for any predicate symbol $P$, $\hat{P} = \hat{P}_0$ (the value
  assigned by $\calM_0$) if $P$ syntactically appears in $A$ and
  $\hat{P} = \hat{P}_1$ (the value assigned by $\calM_1$) otherwise.
\end{itemize}
Let $\interpinb{\_}_i$ be the interpretation generated by $\calM_i$
for $i = 0,1,2$. Then, for any proposition $B$:
\begin{itemize}
\item if $B$ contains only predicate and function symbols appearing in
  $A$,(remind that $\top$ and $\bot$ are connectives),
  $\interpinb{B}_2 = \interpinb{B}_0$  
\item if $B$ contains no predicate or function symbol appearing in
  $A$, $\interpinb{B}_2 = \interpinb{B}_1$
\end{itemize}
\end{lemma}
\begin{proof}
Easy induction on the structure of $B$. The base case is guaranteed
by the definition and it propagates readily. \qed{}
\end{proof}

\begin{theorem} \label{thm:atrans-sc}
Let $R$ be a super-consistent rewrite system and let $A$ be a closed 
$R$-compatible proposition. $\aprop{R}{A}$ is super-consistent.
\end{theorem}
\begin{proof}
Let $\tildeB$ be a pHA. Let $\calM_0$ be any $\tildeB$-valued
structure, and $\interpinb{\_}_0$ the interpretation it generates. Let
$a = \interpinb{A}_0$.

$R$ has a $\asem{\tildeB}{a}$-model because it is
super-consistent. Let $\interpinab{\_}_1$ be the interpretation and
$\calM_1$ the associated $\asem{\tildeB}{a}$-valued
structure. $\calM_1$ is as well a $\tildeB$-valued 
structure, so let $\calM_2$ be the $A$-grafting of $\calM_0$ onto
$\calM_1$, as in \mylem{\ref{lem:struct-comp}}. Let $\interpinb{\_}_2$ and
$\interpinab{\_}_2$ be the interpretations generated in $\tildeB$ and
$\asem{\tildeB}{a}$, respectively. From \mylem{\ref{lem:struct-comp}}
we derive:
\begin{itemize}
\item $\interpinb{A}_2 = \interpinb{A}_0$
\item for any rewrite rule in $R$, $P \rew F$, $\interpinab{P}_2 =
  \interpinab{P}_1$ and $\interpinab{F}_2 = \interpinab{F}_1$
\end{itemize}

In particular, $\interpinab{\_}_2$ inherits from  $\interpinab{\_}_1$
the property to be a model of the rewrite system $R$. We have
fulfilled the requirements of \myprop{\ref{prop:propag-2}}: the pHA is
$\tildeB$, the structure is $\calM_2$, $\interpinab{\_}_2$ is a model
of $R$ in $\asem{\tildeB}{a} = \asem{\tildeB}{\interpinb{A}_2}$, since
$\interpinb{A}_2 = \interpinb{A}_0 = a$.

Therefore $\aprop{R}{A}$ has a $\tildeB$-model for any
$\tildeB$-model, and it is super-consistent.\qed{}
\end{proof}

\section{Super-consistency and Classical Sequent Calculus}

\subsection{From Intuitionistic to Classical Deduction Modulo}

We adapt results of \cite{PNM} to the settings of $A$-translation that
shift cut-elimination in the intuitionistic calculus to the classical
calculus. In the sequel we let $R$ be a rewrite system and $A$ be a
closed $R$-compatible proposition.

\begin{proposition} \label{prop:atrans-rew}
Let $B, C$ be propositions. If $B \rew_R C$ then $\aprop{B}{A}
\rew_{\aprop{R}{A}} \aprop{C}{A}$. If $B \equiv_R C$ then $\aprop{B}{A}
\equiv_{\aprop{R}{A}} \aprop{C}{A}$.
\end{proposition}
\begin{proof}
By induction on the structure of $B$ for the first point, and on the
derivation of $B \equiv_R C$ for the second point. \qed{}
\end{proof}

\begin{proposition}
Assume that $A$ is $R$-compatible. If $R$ is a terminating
and confluent rewrite system\cite{TeReSe03} then so is
$\aprop{R}{A}$.
\end{proposition}
\begin{proof}
Consider a rewriting sequence $A_1 \rew_\aprop{R}{A} \cdots
\rew_\aprop{R}{A} A_n$. $A$ is $R$-compatible, so no proposition or
term appearing in $A$ can be rewritten. Thus we can define the
rewriting sequence $A'_1 \rew_{R} \cdots \rew_{R} A'_n$, starting at
$A'_1 = A_1$ by applying the same rules. This sequence must be
finite.

As for confluence, consider a critical pair $C\,
{}_\aprop{R}{A}\!\leftarrow B \rew_\aprop{R}{A} D$, with $B$ atomic. We
know that $B$ can be rewritten by the corresponding ``antecedent''
rules of $R$: $C_0\,{}_R\!\leftarrow B \rew_R D_0$, with $\aprop{C_0}{A}
= C$ and $\aprop{D_0}{A} = D$. Since $R$ is confluent, there exists
some proposition $E_0$ such that $C_0 \rew^*_R E_0\,{}^*_R\!\leftarrow
D_0$. We also have $C \rew^*_R \aprop{E_0}{A}\, {}^*_R\!\leftarrow D$ by
\myprop{\ref{prop:atrans-rew}}, and $\aprop{R}{A}$ has the diamond
  property \cite{TeReSe03}. Since it is terminating, it is confluent. \qed{}
\end{proof}

\begin{lemma}\label{lem:admissible}
The rules
\begin{tabular}{ccc}
\AXCm{\Gamma, C \vdash A}
\UICm{\Gamma, C \limp A \limp A \vdash A}
\DP & and &
\AXCm{\Gamma \vdash C}
\UICm{\Gamma, C \limp A \vdash A}
\DP
\end{tabular}
are derivable in intuitionistic sequent calculus modulo.
\end{lemma}
\begin{proof}
Direct combination of $\limp$-l, $\limp$-r and axiom rules. \qed{}
\end{proof}

\begin{proposition} \label{prop:clas-to-int}
If the sequent $\G \vdash \D$ has a proof (with cuts) in the classical
sequent calculus modulo $R$ then $\aprop{\G}{A}, (\aprop{\D}{A}) \limp
A \vdash A$ has a proof (with cuts) in the intuitionistic sequent
calculus modulo $\aprop{R}{A}$.
\end{proposition}
\begin{proof}
By an immediate induction we copy the structure of the proof of $\G
\vdash \D$, using \myprop{\ref{prop:atrans-rew}} to rewrite
propositions and the admissible rules of \mylem{\ref{lem:admissible}}
to remove the tail $A$s. This is the only hurdle to get back a sequent
of a shape that allows us to apply the induction hypothesis. 

Notice that, in the $\lor$-r case, we must apply once the
$\lor_1$ rule and once the $\lor_2$, which requires a contraction
on the left-hand side. \qed{}
\end{proof}

\begin{definition}
Let $\Gamma \vdash \Delta; A$ be an intuitionistic sequent. $\Delta$
contains at most one proposition and $\Delta; A$ stands for $A$ if
$\Delta$ is empty and $\Delta$ otherwise.

$\Gamma \vdash \Delta; A$ is said to represent a classical sequent
$A_1, \cdots, A_n \vdash B_1, \cdots, B_p$ if there exists a
one-to-one correspondence $\xi$ between $A_1, \cdots, A_n, B_1, \cdots
B_p$ and $\Gamma, \Delta$:
\begin{itemize}
\item if $\xi(A_i) \in \Gamma$ then $\xi(A_i) = \aprop{A_i}{A}$ or
  $\xi(A_i) = \aprop{A_i}{A} \limp A \limp A$
\item if $\xi(A_i) \in \Delta$ then $\xi(A_i) = \aprop{A_i}{A} \limp A$
\item if $\xi(B_i) \in \Gamma$ then $\xi(B_i) = \aprop{B_i}{A} \limp
  A$
\item if $\xi(B_i) \in \Delta$ then $\xi(B_i) = \aprop{B_i}{A}$ or
  $\xi(B_i) = \aprop{B_i}{A} \limp A \limp A$
\end{itemize}
\end{definition}

\begin{lemma}\label{lem:atrans-form}
Let $B$ be a proposition. Then $\aprop{B}{A}$ cannot be of the forms
$A$, $X \limp A$ and $X \limp A \limp A$.
\end{lemma}
\begin{proof}
A mere check of \mydef{\ref{def:atrans-prop}} according to the
structure of $B$. \qed{}
\end{proof}

\begin{proposition} \label{prop:int-to-clas}
Let $A$ be a proposition. Let $\G \vdash \Delta;
A$ be a sequent that represents $A_1, \cdots, A_n \vdash B_1, \cdots,
B_p$. If this sequent has a cut-free proof in the intuitionistic
sequent calculus modulo $\aprop{R}{A}$, {\em and no right-rule other
  than axiom apply on $A$} then the sequent $A_1, \cdots, A_n \vdash_R
B_1, \cdots, B_p$ has a cut-free proof in the classical sequent
calculus modulo $R$.
\end{proposition}
\begin{proof}
By induction on the intuitionistic proof of the sequent $\G \vdash \D;
A$, using \myprop{\ref{prop:atrans-rew}}:
\begin{itemize}
\item if the last rule is a logical rule applied to a proposition of
  the form $\aprop{A_i}{A}$ or $\aprop{B_i}{A}$, we copy this rule and
  apply the induction hypothesis.
\item If the last rule is a logical rule applied to a proposition of
  another form, it must be an $\limp$-l or a $\limp$-r rule. The
  sequent in the principal premiss is also a representation of the
  sequent $A_1, \cdots, A_n \vdash_R B_1, \cdots, B_p$ - potentially
  weakened by one proposition if $\Delta$ is not empty and a
  $\limp$-l rule was applied. So we just need to apply the induction
  hypothesis, potentially introducing a weak-r if necessary.
\item if the last rule is an axiom, we copy it. Copying an axiom rule
  is possible because, by \mylem{\ref{lem:atrans-form}}, the axiom rule
  can be only applied between propositions of the same nature, with
  no, a single, or two implications with $A$ at the head and the same
  $A$-translated proposition at the base.
\item if the last rule is a structural rule, we copy it on the side
  required by $\xi$ and apply induction hypothesis. \qed{}
\end{itemize}
\end{proof}

It is essential to assume that no rule apply on $A$ other than axiom,
otherwise the result fails; for instance the sequent $\vdash ; C \limp
C$ is intuitionistically provable while the empty sequent is not
classically provable.

\subsection{Cut Elimination in Classical Sequent Calculus Modulo}

\begin{theorem} \label{thm:sc-cutelim}
If a rewrite system $R$ is super-consistent the classical sequent
calculus modulo $R$ has the cut elimination property.
\end{theorem}
\begin{proof}
Let $\G \vdash \D$ be a provable sequent in the classical sequent
calculus modulo $R$. Let $A$ be a proposition not containing any
predicate or function symbol of $R$. The sequent $\aprop{\G}{A},
\aprop{\D}{A} \limp A \vdash A$ has a proof in the intuitionistic
sequent calculus modulo $\aprop{R}{A}$ by
\myprop{\ref{prop:clas-to-int}} above. By \mythm{\ref{thm:atrans-sc}}, 
$\aprop{R}{A}$ is super-consistent. Therefore, by \mycor{4.1}
of \myprop{4.1} of, $\aprop{\G}{A}, \aprop{\D}{A} \limp A
\vdash A$ has a cut-free proof in the intuitionistic sequent
calculus.

Moreover, no rule on $A$ other than axiom is introduced:
\myprop{\ref{prop:clas-to-int}} introduces only axioms, that are
translated into axioms in natural deduction, and the structure of $A$
is therefore not exposed to any introduction or elimination
rules. Another argument is that we can ``freeze'' $A$ and view it
as an atomic formula in all the discussion above. So the proof cannot
use any information on $A$, since it is a generic parameter of the
theorem.

Consequently, by \myprop{\ref{prop:int-to-clas}} the sequent $\G
\vdash \D$ has a cut-free proof. \qed{}
\end{proof}

Note that the argument appeals to a normalization procedure of the
proof-terms of Natural deduction modulo, considering commutative cuts
(\mysec{3.6} of \cite{PNM}). Other cut elimination methods for Natural
deduction modulo (as the one of \cite{GDowOHer12}) do not apply since
they do not get rid of commutative cuts.

\section{Conclusion}

In \cite{PNM} $\aprop{R}{\bot}$ had to be
assumed to have a pre-model in order to show cut elimination for the
classical sequent calculus modulo $R$ (\mythm{4.1} of
\cite{PNM}). \cite{TVA} shows that it is sufficient to show
$\aprop{R}{\bot}$ to be super-consistent. We have shown here that we
can instead discuss the super-consistency of $R$ directly.

Our result is a priori more restrictive, since by
instantiating $A$ by $\bot$ we get the super-consistency of
$\aprop{R}{\bot}$ that in turn implies the existence of a pre-model
for $\aprop{R}{\bot}$. It is currently unknown whether all those
criteria are equivalent or not: can we, for instance, find a rewrite
system and a proposition $A$, such that $\aprop{R}{A}$ is
super-consistent while $R$ is not super consistent ? Does the
existence of a pre-model for $R$ entail super-consistency ? On the
good side, our criterion works directly on $R$ and avoids a
duplication of arguments: we now in one pass have normalization for
natural deduction modulo $R$ (\cite{TVA,PNM}) and cut elimination for
the classical sequent calculus, and bypass the need of two separate
pre-model (or super-consistency arguments) for $R$ and
$\aprop{R}{\bot}$. Moreover, super-consistency, by abstracting over
reducibility candidates, provides a certain ease of use.

We have also shown a general result, by $A$-translating rewrite
systems and semantics frameworks, instead of $\bot$-translating
them. For the proof of cut elimination, we believe that the latter,
better known as double-negation translation, would have been
sufficient, as in \cite{PNM}. But the work on $A$-translation bears a
more general character, that can be used for other applications.

Super-consistency appears to be the right criterion to deal
with when one wants to know about the cut elimination property of a
deduction modulo theory, as the property holds whatever the syntactic
calculus is. It would be interesting to see how the super-consistency
criterion extends to other first-order framework, like the calculus of
structures \cite{GUG} or $\lambda\Pi$-calculus modulo, that is at the
root of the {\rm Dedukti} proof-checker \cite{BA12}.

Whether we can widen the criterion and replace pseudo-Heyting
algebras by Heyting algebras in \mydef{\ref{def:sc}}, the idea being
to use {\em cut-admissibility} (through semantic completeness, in the
mood of \cite{OHerJLip08b} for instance) instead of normalization in
the proof of \mythm{\ref{thm:sc-cutelim}} is a conjecture. Analyzing 
\cite{ABruOHerCHou11,GDowOHer12} closely shows that cut-admissibility
results crucially depend on finding in the
interpretation of the atoms $P$ a syntactical version of
$P$ in the model formed out of
contexts/propositions. Super-consistency does not {\em directly}
allows this, due to the abstract construction of a generic model. This
appeals to a more informative structure, in both papers algebras of
sequents were introduced which happens to be only pseudo-Heyting
algebras.



\bibliographystyle{plain} 
\bibliography{SemanticTranslation}

\begin{thebibliography}{10}

\bibitem{BA12}
Mathieu Boespflug, Quentin Carbonneaux, and Olivier Hermant.
\newblock {The $\lambda\Pi$-Calculus Modulo as a Universal Proof Language}.
\newblock In {\em Proof Exchange for Theorem Proving (PxTP)}, pages 28--43,
  Manchester (UK), June 2012.

\bibitem{RB04}
Richard Bonichon.
\newblock {TaMeD: A Tableau Method for Deduction Modulo}.
\newblock In {\em International Joint Conference on Automated Reasoning
  (IJCAR)}, volume 3097 of {\em LNCS}, pages 445--459, Cork (Ireland), July
  2004. Springer.

\bibitem{BH06a}
Richard Bonichon and Olivier Hermant.
\newblock A semantic completeness proof for tableaux modulo.
\newblock {\em LPAR 2006}, pages 167--181, November 2006.

\bibitem{ABruOHerCHou11}
Alo\"{\i}s Brunel, Olivier Hermant, and Cl{\'e}ment Houtmann.
\newblock Orthogonality and boolean algebras for deduction modulo.
\newblock In C.-H.~Luke Ong, editor, {\em TLCA}, volume 6690 of {\em Lecture
  Notes in Computer Science}, pages 76--90. Springer, 2011.

\bibitem{GBur10}
Guillaume Burel.
\newblock Embedding deduction modulo into a prover.
\newblock In Anuj Dawar and Helmut Veith, editors, {\em CSL}, volume 6247 of
  {\em Lecture Notes in Computer Science}, pages 155--169. Springer, 2010.

\bibitem{GBurCKir10}
Guillaume Burel and Claude Kirchner.
\newblock Regaining cut admissibility in deduction modulo using abstract
  completion.
\newblock {\em Inf. Comput.}, 208(2):140--164, 2010.

\bibitem{TVA}
Gilles Dowek.
\newblock Truth values algebras and proof normalization.
\newblock In Thorsten Altenkirch and Conor McBride, editors, {\em TYPES},
  volume 4502 of {\em Lecture Notes in Computer Science}, pages 110--124.
  Springer, 2006.

\bibitem{HOLls}
Gilles Dowek, Th{\'e}r{\`e}se Hardin, and Claude Kirchner.
\newblock {HOL-$\lambda \sigma$} an intentional first-order expression of
  higher-order logic.
\newblock {\em Mathematical Structures in Computer Science}, 11(1):21--45,
  2001.

\bibitem{GDowOHer12}
Gilles Dowek and Olivier Hermant.
\newblock A simple proof that super-consistency implies cut elimination.
\newblock {\em Notre-Dame Journal of Formal Logic}, 53(4):439--456, 2012.

\bibitem{PNM}
Gilles Dowek and Benjamin Werner.
\newblock Proof normalization modulo.
\newblock {\em The Journal of Symbolic Logic}, 68(4):1289--1316, 2003.

\bibitem{ARITH}
Gilles Dowek and Benjamin Werner.
\newblock Arithmetic as a theory modulo.
\newblock In J{\"u}rgen Giesl, editor, {\em RTA}, volume 3467 of {\em Lecture
  Notes in Computer Science}, pages 423--437. Springer, 2005.

\bibitem{FRIEDMAN}
Harvey Friedman.
\newblock Classically and intuitionistically provably recursive functions.
\newblock In Gert~H. M{\"{u}}ller and Dana~S. Scott, editors, {\em Higher Set
  Theory}, volume 669 of {\em Lecture Notes in Mathematics}, pages 21--27.
  Springer Berlin Heidelberg, 1978.

\bibitem{GEN}
Gerhard Gentzen.
\newblock Die widerspruchsfreiheit der reinen zahlentheorie.
\newblock {\em Mathematische Annalen}, 112:493:565, 1936.

\bibitem{GOD}
Kurt G{\"{o}}del.
\newblock Zur intuitionistischen arithmetik und zahlentheorie.
\newblock {\em Ergebnisse eines mathematischen Kolloquiums}, 4:34--38, 1933.

\bibitem{GUG}
Alessio Guglielmi.
\newblock A system of interaction and structure.
\newblock {\em ACM Trans. Comput. Log.}, 8(1):1--64, 2007.

\bibitem{OHer05a}
Olivier Hermant.
\newblock Semantic cut elimination in the intuitionistic sequent calculus.
\newblock In Pawel Urzyczyn, editor, {\em Typed Lambda-Calculi and
  Applications}, volume 3461 of {\em LNCS}, pages 221--233, Nara, Japan, 2005.
  Springer.

\bibitem{OHerJLip08b}
Olivier Hermant and James Lipton.
\newblock A constructive semantic approach to cut elimination in type theories
  with axioms.
\newblock In Michael Kaminski and Simone Martini, editors, {\em CSL}, volume
  5213 of {\em Lecture Notes in Computer Science}, pages 169--183. Springer,
  2008.

\bibitem{JA12}
M\'elanie Jacquel, Karim Berkani, David Delahaye, and Catherine Dubois.
\newblock {Tableaux Modulo Theories using Superdeduction: An Application to the
  Verification of {\textsf{B}} Proof Rules with the {\textsf{Zenon}} Automated
  Theorem Prover}.
\newblock In {\em International Joint Conference on Automated Reasoning
  (IJCAR)}, volume 7364 of {\em LNCS}, pages 332--338, Manchester (UK), June
  2012. Springer.

\bibitem{TeReSe03}
TeReSe.
\newblock {\em Term Rewriting Systems}, volume~55 of {\em Cambridge Tracts in
  Theoretical Computer Science}.
\newblock Cambridge University Press, 2003.

\bibitem{ATroDvDal88}
Anne~Sjerp Troelstra and Dirk van Dalen.
\newblock {\em Constructivism in Mathematics, An Introduction}.
\newblock North-Holland, 1988.

\end{thebibliography}
\end{document}